\documentclass[letterpaper, 10pt, conference]{ieeeconf}

\usepackage{amssymb}
\usepackage{amsmath}
\usepackage{color}
\usepackage{cite}
\usepackage{graphicx}
\usepackage{url}

\IEEEoverridecommandlockouts

\newtheorem{thm}{Theorem}
\newtheorem{lem}{Lemma}
\newtheorem{mydef}{Definition}

\newtheorem{ex}{Example}
\newtheorem{asm}{Assumption}
\newtheorem{prop}{Proposition}

\title{\LARGE \bf Optimal Budget Allocation in Social Networks: Quality or Seeding?
}

\author{Arastoo Fazeli$^\dagger$  \quad Amir Ajorlou $^\dagger$  \quad Ali Jadbabaie$^\dagger$
\thanks{\noindent$^\dagger$Department of Electrical and Systems Engineering and GRASP Laboratory at University of Pennsylvania. {\tt\small arastoo@seas.upenn.edu}, {\tt\small ajorlou@seas.upenn.edu} and {\tt\small jadbabai@seas.upenn.edu.} This research was supported by ARO MURI W911NF-12-1-0509 and AFOSR Complex Networks Program.}
}

\begin{document}

\maketitle
\thispagestyle{empty}
\pagestyle{empty}


\begin{abstract}

In this paper, we study a strategic model of marketing and product consumption in social networks. We consider two competing firms in a market providing two substitutable products with preset qualities. Agents choose their consumptions following a myopic best response dynamics which results in a local, linear update for the consumptions. At some point in time, firms receive a limited budget which they can use to trigger a larger consumption of their products in the network. Firms have to decide between marginally improving the quality of their products and giving free offers to a chosen set of agents in the network in order to better facilitate spreading their products. We derive a simple threshold rule for the optimal allocation of the budget and describe the resulting Nash equilibrium. It is shown that the optimal allocation of the budget depends on the entire distribution of centralities in the network, quality of products and the model parameters. In particular, we show that in a graph with a higher number of agents with centralities above a certain threshold, firms spend more budget on seeding in the optimal allocation. Furthermore, if seeding budget is nonzero for a balanced graph, it will also be nonzero for any other graph, and if seeding budget is zero for a star graph, it will be zero for any other graph too. We also show that firms allocate more budget to quality improvement when their qualities are close, in order to distance themselves from the rival firm. However, as the gap between qualities widens, competition in qualities becomes less effective and firms spend more budget on seeding.

\end{abstract}


\section{Introduction}

Many recent studies have documented the role of social networks in individual purchasing decisions
\cite{feick1987market,reingen1984brand,godes2004using}. More data from online social networks and advances in information technologies have drawn the attention of firms to exploit this information for their marketing goals. As a result, firms have become more interested in models of influence spread in social networks in order to improve their marketing strategies. In particular, considering the relationship between people in social networks and their rational choices, many retailers are interested to know how to use the information about the dynamics of the spread in order to maximize their product consumption and achieve the most profit in a competitive market.

A main feature of product consumption in these settings is what is often called the ``network effect'' or positive externality. For such products, consumption of each agent incentivizes the neighboring agents to consume more as well, as the consumption decisions between agents and their neighbors are strategic complements of each other.
There are diverse sets of examples for such products or services. New technologies and innovations (e.g., cell phones), network goods and services (e.g., fax machines, email accounts), online games (e.g., Warcraft), social network web sites (e.g., Facebook, Twitter) and online dating services (e.g., OkCupid) are among many examples in which consuming from a common product or service is more preferable for people. Also, a main property of many products is substitution. A substitute product is a product or service that satisfies the need of a consumer that another product or service fulfills (e.g. Pepsi and Coke or email and fax). In all these examples, firms might be interested to take advantage of the social network among people and the positive externality of their products and services to trigger a larger consumption of their products. Therefore, it might be important for firms to know how to shape their strategies in designing their products and offering them to a set of people in order to promote their products intelligently, and eventually achieving a larger share of the market.

In this paper, we study strategic competition between two firms trying to maximize their product consumption. Firms simultaneously allocate their fixed budgets between seeding a set of costumers embedded in a social network and improving the quality of their products. The consumption of each agent is the result of its myopic best response to the previous actions of its peers in the network. Therefore, a firm should provide enough incentives for spread of its product through the payoff that agents receive by consuming it. For this purpose and considering their budgets, firms should strategically design their products and know how to initially seed the network.

We model the above problem as a fixed-sum game between firms, where each firm tries to maximize discounted sum of its product consumption (i.e. market share) over time, considering its fixed budget.
We derive a simple rule for optimal allocation of the budget between improving the quality and new seeding which depends on the network structure, quality of products, and parameters of the model. It is shown that the optimal allocation of the budget depends on the entire centrality distribution of the graph. Specially, we show that maximum seeding occurs in a graph with maximum number of agents with centralities above a certain threshold.
Also, the difference in qualities of firms plays an important role in the optimal allocation of the budget. In particular, we show that as the gap between the qualities of the products widens, the firms allocate more budget to seeding.

It is worthwhile to note that the problem of influence and spread in networks has been extensively studied in the past few years \cite{ballester2006s,bharathi2007competitive,galeotti2009influencing,kempe2003maximizing,kempe2005influential,chasparis2010control,vetta2002nash}. Also, diffusion of new behaviors and strategies through local coordination games has been an active field of research
\cite{ellison1993learning,kandori1993learning,harsanyi1988general,young1993evolution,young2001individual,young2002diffusion,montanari2010spread,kleinberg2007cascading}. Goyal and Kearns proposed a game theoretic model of product adoption in \cite{goyal2012competitive}. They computed upper bounds of the price of anarchy and showed how network structure may amplify the initial budget differences. Similarly, in \cite{bimpikiscompeting} Bimpikis, Ozdaglar and Yildiz proposed a game theoretic model of competition between firms which can target their marketing budgets to individuals embedded in a social network. They provided conditions under which it is optimal for the firms to asymmetrically target a subset of the individuals. Also, Chasparis and Shamma assumed a dynamical model of preferences in \cite{chasparis2010control} and computed optimal policies for finite and infinite horizon where endogenous network influences, competition between two firms and uncertainties in the network model were studied.
The main contribution of our work is to explicitly study the tradeoff between investing on improving the quality of a product and initial seeding in a social network. Our model is similar to the model proposed in \cite{fazeli2012duopoly}, however, instead of pricing strategy in \cite{fazeli2012duopoly}, the notion of quality is introduced and the tradeoff between quality improvement and seeding is studied. Also, our model is tractable and allows us to characterize the exact product consumption at each time, instead of lower and upper bounds provided in \cite{fazeli2012game,fazeli2012targeted}.

The rest of this paper is organized as follows: In Section~\ref{sec2}, we introduce our model and update dynamics for agents applying the myopic best response. In Section~\ref{sec3}, we study the game played among the firms. Optimal budget allocation strategy is presented in Section~\ref{sec4}. Finally, in Section~\ref{sec5}, we conclude the paper.


\section{The Spread Dynamics} \label{sec2}

There are $n$ agents $V=\{1,\ldots,n\}$ in a social network. The relationship among agents is represented by a directed graph $\mathcal{G}=(V,E)$ in which agents $i, j \in V$ are neighbors if $(i,j) \in E$. The weighted adjacency matrix of the graph $\mathcal{G}$ is denoted by a row stochastic matrix $G$ where the $ij$-th entry of $G$, denoted by $g_{ij}$, represents the strength of the influence of agent $j$ on $i$. For diagonal elements of matrix $G$, we have $g_{ii}=0$ for all agents $i \in V$. We assume that there are two competing firms $a$ and $b$ producing product $a$ and $b$. Each agent has  one unit demand which can be supplied by either of the firms.
We define the variable $0 \leq x_i(t) \leq 1$ and $0 \leq 1- x_i(t) \leq 1$ as the consumption of the product $a$ and $b$ by agent $i$ at time $t$.

Denote by $q_a > 0$ and $q_b >0 $ the quality of product $a$ and $b$ respectively. The values of $q_a$ and $q_b$ can be interpreted as the payoff that any two agents $i$ and $j$ would achieve if they both consume the same product. In other words, we can assume $q_a$ and $q_b$ are payoffs of the following game
\begin{center}
\begin{tabular}{|l|c|r|}
  \hline
   & $x_j$ & $1-x_j$ \\ \hline
  $x_i$ & $q_a x_i x_j$ & $0$ \\ \hline
  $1-x_i$ & $0$ & $q_b (1-x_i)(1-x_j)$ \\
  \hline
\end{tabular}
\end{center}
Since agents benefit from the same action of their neighbors, this game could be thought of as a local coordination game. From the above table it follows easily that the payoff of agents $i$ and $j$ from their interaction is
\begin{equation*}
u_{ij}(x_i,x_j)=q_a x_i x_j + q_b (1-x_i)(1-x_j).
\end{equation*}
We also assume that each agent benefits from taking action $x_i$ irrespective of actions taken by its neighbors. We assume the isolation payoff of consuming $x_i$ and $1-x_i$ from product $a$ and $b$ is represented by the following quadratic form functions
\begin{equation*}
u_{ii}^{a} = q_a (x_i - x^2_i), \qquad  u_{ii}^{b} = q_b [(1-x_i) - (1-x_i)^2].
\end{equation*}
This forms of payoff indicates that a product with higher quality has a higher isolation payoff. The total isolation payoff of agent $i$ can be written as
\begin{equation*}
u_{ii}(x_i) = \{ q_a (x_i - x^2_i)\} + \{q_b [(1-x_i) - (1-x_i)^2]\}.
\end{equation*}
The above isolation payoff indicates that the isolation utility of agent $i$ disappears when $x_i=0$ or $x_i=1$, i.e. the agent consumes only one of the product $a$ or $b$. Assuming quadratic form function for the isolation payoff not only makes the analysis more tractable, but also is a good second order approximation for the general class of concave payoff functions. By changing the variables $x_i = \frac{1}{2} + y_i$ we get
\begin{equation*}
\begin{split}
&u_{ij}(y_i,y_j)=q_a (\frac{1}{2} + y_i) (\frac{1}{2} + y_j) + q_b (\frac{1}{2} - y_i)(\frac{1}{2} - y_j), \\
&u_{ii}(y_i) = (q_a + q_b)(\frac{1}{4} - y^2_i).
\end{split}
\end{equation*}
Therefore, the total utility of agent $i$ from taking action $y_i$ is given by
\begin{equation} \label{total utility}
\begin{split}
&U_i(y_i, \vec{y}_{-i}) = \alpha \{ (q_a + q_b)(\frac{1}{4} - y^2_i) \} + (1-\alpha) \\
&\{q_a \sum_{j=1}^n g_{ij} (\frac{1}{2} + y_i) (\frac{1}{2} + y_j) + q_b \sum_{j=1}^n g_{ij} (\frac{1}{2} - y_i) (\frac{1}{2} - y_j)\},
\end{split}
\end{equation}
where $0 \leq \alpha \leq 1$ represents the weight that each agent puts on its isolation payoff compared to the payoff from its neighbors. In the above equation $\vec{y}_{-i}$ denotes an action vector of all agents other than agent $i$. From equation \eqref{total utility} we can see that product $a$ and $b$ have a positive externality effect in the network, meaning that the usage level of an agent has a positive impact on the usage level of its neighbors. Therefore, it follows that $q_a$ and $q_b$ in addition to the payoff of a local coordination game, can be interpreted as coefficients of network externality of product $a$ and $b$ respectively.

We assume agents repeatedly apply myopic best response to the actions of their neighbors. This means that each agent, considering its neighbors consumptions at the current period, chooses the amount of the product that maximizes its current payoff, as its consumption for the next period. In other words, consumption of agent $i$ at time $t+1$ is updated as follows
\begin{equation*} \label{utility}
y_i(t+1) = \arg\max_{y_i}\quad U_i(y_i(t), \vec{y}_{-i}(t)).
\end{equation*}
The above equation results in the following update dynamics
\begin{equation*} \label{best response}
y_i(t+1)= (\frac{1-\alpha}{2\alpha})\sum_{j=1}^n g_{ij} y_j(t) + (\frac{(1-\alpha)(q_a-q_b)}{4\alpha(q_a+q_b)})\sum_{j=1}^n g_{ij}.
\end{equation*}
Therefore, the consumption of the product $a$ can be written as the following linear update dynamics form
\begin{equation} \label{update a}
y(t+1)= (\frac{1-\alpha}{2\alpha}) G y(t) + (\frac{(1-\alpha)(q_a-q_b)}{4\alpha(q_a+q_b)}) \mathbf{1}.
\end{equation}
Similarly, for the consumption of the product $b$ we have \\
$1 - x_i(t) = \frac{1}{2} - y_i(t)$.
\begin{asm} \label{asm1}
We have $\frac{1}{2} \leq \alpha$, i.e. the self isolated payoff has higher weight than the payoff from neighbors. This assumption guaranties that $0 \leq x_i(t) \leq 1$ for all $i$ and all $t$ under the update rule \eqref{update a}.
\end{asm}
Using the above assumption \label{asm1} and defining
\begin{equation} \label{definition}
W \triangleq (\frac{1-\alpha}{2\alpha}) G, \quad u_a \triangleq (\frac{(1-\alpha)(q_a-q_b)}{4\alpha(q_a+q_b)}) \mathbf{1},
\end{equation}
equations \eqref{update a} can be written as
\begin{equation} \label{matrix form}
y(t+1)= W y(t) + u_a.
\end{equation}
Equation \eqref{matrix form} can be expanded as
\begin{equation} \label{expanded form}
y(t)= W^t y(0) + \sum_{k=0}^{t-1} W^k u_a.
\end{equation}
Therefore, the consumption of agents depends on the initial preference, i.e. $y(0)$, the quality of product $a$ and $b$, i.e. $q_a$ and $q_b$, and the structure of the network, i.e. the matrix $G$. In the next section we discuss how firms can exploit this information in order to maximize the consumption of their products and also characterize a Nash equilibrium of the game played between two firms.

\section{The Game Between Firms} \label{sec3}

In this section we describe the game between firms where each firm maximizes the spread of its product given a fixed budget. As we mentioned earlier, firms produce their products with some preset quality. At some point in time, say $t=0$, firms receive a fixed budget that they can either invest on improving the quality of their products around the preset value or spend it on new seedings of some agents or both. This new seeding can be viewed as free offer to promote their products in the network.
We define the utility of each firm as the discounted sum of its product consumption over time
\begin{equation*}
\begin{split}
&U_a(q_a,q_b) = \sum_{t=0}^{\infty} \delta^t \mathbf{1}^T (\frac{\mathbf{1}}{2} + y(t)), \\
&U_b(q_a,q_b)= \sum_{t=0}^{\infty} \delta^t \mathbf{1}^T (\frac{\mathbf{1}}{2} - y(t)).
\end{split}
\end{equation*}
As it can be seen from equations above, firms play a fixed-sum game. Using equations \eqref{definition} and \eqref{expanded form} and defining the centrality vector $v$ by $v= (I-\delta W^T)^{-1} \mathbf{1}$
and noting that $\sum v_i = \frac{2\alpha n}{2\alpha- \delta(1-\alpha)}$,
the utilities of firms can be written as
\begin{equation} \label{payoffs}
\begin{split}
& U_a(q_a,q_b) = (\frac{n}{2(1-\delta)}) + v^T y(0) + \lambda (\frac{q_a-q_b}{q_a+q_b}), \\
& U_b(q_a,q_b) = (\frac{n}{2(1-\delta)}) - v^T y(0) - \lambda (\frac{q_a-q_b}{q_a+q_b}),
\end{split}
\end{equation}
where
\begin{equation} \label{lambda}
\lambda = \frac{\delta(1-\alpha) n}{2(1-\delta)(2\alpha- \delta(1-\alpha))}.
\end{equation}
We assume the cost of improving the quality by $\Delta q$ is given by $c_q \Delta q$ where $c_q$ is a large number and the cost of each unit of new seeding is given by $c_s$. Each firm has a limited budget that can spend on either new seeding, i.e. $S_a$ and $S_b$, or enhancing the quality of its product, i.e. $\Delta q_a$ and $\Delta q_b$, or both. In order to have $0 \leq x_i(0) \leq 1$ and $0 \leq 1 - x_i(0) \leq 1$ for all agents $i$, we impose the constraints $\|S_a+y(0)\|_\infty\leq 0.5$ and $\|S_b-y(0)\|_\infty\leq 0.5$. This means that firms can initially seed agents up to their demand capacity. Seeding $S_a$ and $S_b$ will change the initial consumption of products $a$ and $b$ by $S_a-S_b$ and $S_b-S_a$ respectively. From equation \eqref{payoffs} the marginal change in the payoff of firm $a$ and $b$ resulted from the new investment are given by
\begin{equation*}
\begin{split}
&\Delta U_a=v^TS_a - v^TS_b + \frac{2\lambda q_b\Delta q_a}{(q_a+q_b)^2} - \frac{2\lambda q_a\Delta q_b}{(q_a+q_b)^2}, \\
&\Delta U_b=v^TS_b - v^TS_a + \frac{2\lambda q_a\Delta q_b}{(q_a+q_b)^2} - \frac{2\lambda q_b\Delta q_a}{(q_a+q_b)^2}.
\end{split}
\end{equation*}
Each firm maximizes its utility, or equivalently its marginal payoff, given its fixed budget. Since the effect of the action of firm $b$, i.e. $S_b$ and $\Delta q_b$, is decoupled from that of the action of firm $a$ in $\Delta U_a$, thus firm $a$ should solve the following optimization problem
\begin{equation}  \label{firms utility function a}
\begin{split}
& \max_{S_a, \Delta q_a} \quad v^TS_a+\frac{2\lambda q_b\Delta q_a}{(q_a+q_b)^2}, \\
& \text{Subject to} \quad c_s \|S_a\|_1 + c_q \Delta q_a=K_a.
\end{split}
\end{equation}
Similarly, for the firm $b$ we have
\begin{equation}  \label{firms utility function b}
\begin{split}
& \max_{S_b, \Delta q_b} \quad v^TS_b+\frac{2\lambda q_a\Delta q_b}{(q_a+q_b)^2}, \\
& \text{Subject to} \quad c_s \|S_b\|_1 + c_q \Delta q_b=K_b.
\end{split}
\end{equation}
From equations \eqref{firms utility function a} and \eqref{firms utility function b} it can be seen that the optimal strategy of each firm is independent of the action of the other firm. This results in a Nash equilibrium to be simply the pair of the optimal actions of the firms. In the next section we discuss the optimal strategy of each firm and the resulting Nash equilibrium.

\section{Optimal Budget Allocation Strategy} \label{sec4}

In this section we study how firms can maximize the spread of their products by optimally allocating their fixed budget between improving the quality of their products and new seeding.  We also describe a Nash equilibrium of the game played between these two firms described in Section \ref{sec3}. The next theorem describes a simple rule for the optimal allocation of the budget.
\begin{thm}
For firm $a$, it is more profitable to seed agent $j$ rather than enhancing the quality of its product if $v_j > v_c^a$ where
\begin{equation} \label{condition}
v_c^a \triangleq (2\lambda) (\frac{c_s}{c_q}) (\frac{q_b}{(q_a+q_b)^2}).
\end{equation}
Similarly, for firm $b$, it is more profitable to seed agent $j$ rather than enhancing the quality of its product if $v_j > v_c^b$ where
\begin{equation} \label{condition2}
v_c^b \triangleq (2\lambda) (\frac{c_s}{c_q}) (\frac{q_a}{(q_a+q_b)^2}).
\end{equation}
Moreover, any pair of the optimal strategies of the firms described by the above threshold rules describes a Nash equilibrium.
\end{thm}
\begin{proof}
From equation \eqref{condition} the relative marginal utility to cost for spending budget to seed agent $j$ is $\frac{v_j}{c_s}$. Therefore, it is always more profitable to seed an agent with higher centrality. Also, the relative marginal utility to cost for spending budget on enhancing quality of product $a$ is $\frac{2 \lambda q_b}{c_q(q_a+q_b)^2}$. Therefore, for firm $a$ it is more profitable to seed agent $j$ rather than enhancing the quality of its product iff
\begin{equation*}
\frac{v_j}{c_s} > \frac{2 \lambda q_b}{c_q(q_a+q_b)^2}.
\end{equation*}
This completes the proof. Similar story holds for firm $b$. Moreover, since the best response of each firm resulting from equations \eqref{firms utility function a} and \eqref{firms utility function b} is independent of the action of the other firm, any Nash equilibrium of the game between firms is simply a pair of firms best responses.
\end{proof}

Following the above theorem, the optimal allocation of the budget for each firm is to follow a so called water-filling strategy, that is, to start seeding in the order of agents' centralities until the centrality falls below the threshold given by \eqref{condition} for firm $a$ or \eqref{condition2} for firm $b$ (in which case the firm spends the rest of the budget on improving the quality) or the firm runs out of budget. Also, the amount that agents can be seeded is up to their demand capacity, i.e. $S_a^{max} = (0.5)\mathbf{1} - y(0) > 0$ and $S_b^{max} = (0.5)\mathbf{1} + y(0) > 0$. Also, note that if the centrality of any agent is equal to the threshold defined in \eqref{condition} or \eqref{condition2},
then firms are indifferent between seeding that agent and quality improvement. Equations \eqref{condition} and \eqref{condition2} indicate that the optimal allocation depends on quality of products, i.e. $q_a$ and $q_b$, centrality distribution of agents in the network, i.e. $v$, the parameter $\lambda$ which depends on discounting factor $\delta$ and network coefficient $\alpha$, and cost of seeding and quality improvement, i.e. $c_s$ and $c_q$. We will discuss the effect of each of these factors on optimal allocation in the following subsections. For simplicity, we only discuss optimal seeding budget; optimal quality improvement budget can be found easily using the budget constraint $c_s \|S_a\|_1 + c_q \Delta q_a=K_a$. All of our analysis here is for firm $a$ and similar results can be shown for firm $b$ as well.
\subsection{Effect of Network Structure on Firms' Decisions:}

In this subsection we study the effect of network structure on the optimal allocation of the budget for seeding and quality improvement. First we define seeding capacity of a graph.
\begin{mydef}
The seeding capacity of a graph is the amount that it can be seeded in the optimal allocation when there is no budget constraint.
\end{mydef}

We first focus on two well studied graphs, i.e. star and balanced graphs, and highlight how they can reflect  important properties of the seeding budget. Before continuing further, we find the network centralities for these two graphs in the next lemma.
\begin{lem}\label{v_Lv_H}
The centrality of the agents in a balanced graph is given by $\bar{v}=\frac{2\alpha}{2\alpha-\delta(1-\alpha)}$. In a star graph, the centrality of the central agent is
\begin{equation*}
v_h=\frac{1+\frac{\delta(1-\alpha)(n-1)}{2\alpha}}{1-(\frac{\delta(1-\alpha)}{2\alpha})^2},
\end{equation*}
and the centrality of non central agents is
\begin{equation*}
v_l=\frac{1+\frac{\delta(1-\alpha)}{2\alpha(n-1)}}{1-(\frac{\delta(1-\alpha)}{2\alpha})^2}.
\end{equation*}
Moreover, for any arbitrary graph $G$, $\bar{v}\leq v_{max}\leq v_h$, where $v_{max}=\max_{i\in V}v_i$.
\end{lem}
\begin{proof}
First part simply follows from the fact that $v=(I-\delta W^T)^{-1}\mathbf{1}$, where $W$ is given by \eqref{definition}, and that for any agent $i$ in a balanced graph $\sum g_{ji}=\sum g_{ij}=1$. For the star graph, noting that $v=\mathbf{1}+\delta W^T v$, we can obtain
\begin{align*}
v_h&=1+\delta(n-1)v_l,\\
v_l&=1+\frac{\delta v_h}{(n-1)},
\end{align*}
solving which we can find $v_h$ and $v_l$ as given in the lemma.

Also, for any arbitrary graph $G$, $v_{max}\geq\frac{\sum v_i}{n}=\bar{v}$. To show $v_{max}\leq v_h$, assume that the maximum centrality occurs for agent $i$. Again, using $v=\mathbf{1}+\delta W^T v$ we can obtain
\begin{equation*}
v_j\geq 1+\frac{\delta(1-\alpha)}{2\alpha} g_{ij} v_i,
\end{equation*}
using which for all $j\neq i$, we get
\begin{equation*}
\sum_{j=1}^{n}{v_j}\geq (n-1)+(1+\frac{\delta(1-\alpha)}{2\alpha})v_i.
\end{equation*}
Applying simple algebra along with the fact that $\sum v_j=\frac{2\alpha n}{2\alpha-\delta(1-\alpha)}$ leads to $v_i\leq v_h$.
\end{proof}

The next proposition provides a condition for seeding profitability of any general graph. Also, the seeding capacity of star and balanced graphs are compared and it is shown that the graph with higher seeding capacity can be any of the two, depending on the threshold value of $v_c^{a}$ in \eqref{condition}.
\begin{prop} \label{compare seeding}
If seeding capacity is nonzero for a balanced graph, it will be nonzero for any other graph too. On the other hand, if seeding capacity is zero for a star graph, it will also be zero for any other graph. Moreover, if $ v_l < v_c^a < \bar{v}$, a balanced graph has a larger seeding capacity than a star graph, and if $\bar{v}< v_c^a <  v_h$, a star graph has a larger seeding capacity than a balanced graph. For $1<v_c^a<v_l$ they have the same seeding capacity.
\end{prop}
\begin{proof}
If seeding capacity is nonzero for a balanced graph, then we have  $v_c^a < \bar{v}$. As a result, for any other graph we will have $v_c^a <v_{max}$, since according to Lemma~\ref{v_Lv_H} $\bar{v}\leq v_{max}$. This means that there exists at least one agent that must be seeded. On the other hand, if seeding capacity is zero for a star graph, then we must have $v_c^a > v_h$. Since we know $v_h \geq v_i$ for any agent $i$ of any arbitrary graph, therefore, $v_c^a > v_i$ and no agent can be seeded in any other graph.
For the second part of the proposition,
if $ v_l < v_c^a < \bar{v}$, then seeding capacity for the star graph will be $S^{max}_{a_i}$, where $S^{max}_{a_1} \geq S^{max}_{a_2} \geq \cdots \geq S^{max}_{a_n}$ are elements of the demand capacity vector $S^{max}_a$ and agent $i$ is the central agent. However, for the balanced graph all agents can be seeded up to their maximum demand capacities and the seeding capacity will be $\|S^{max}_a\|_1$. On the other hand, if $ \bar v < v_c^a < v_h$, still seeding for the star graph will be $S^{max}_{a_i}$, however, no agent can be seeded in the balanced graph. For $1<v_c^a<v_l$, agents in both graphs can be seeded up to $\|S^{max}_a\|_1$.
\end{proof}

Although star and balanced graphs give us useful insights for seeding capacity in general graphs, extreme seeding capacities do not necessary happen in those graphs. The next proposition provides us with a lower and an upper bounds for minimum and maximum seeding capacities.
\begin{prop} \label{max seeding}
If $ 1 < v_c^a < v_h$, the maximum seeding capacity is given by  \begin{equation*}
\|S^{*}_a\|^{max}_{1} = \sum_{i=1}^{k} S^{max}_{a_i},
\end{equation*}
where
\begin{equation} \label{k}
k = \min \{\lfloor \frac{n\delta (1-\alpha) }{(v_c^a-1)(2\alpha - \delta(1-\alpha))} \rfloor, n\}.
\end{equation}
On the other hand, the minimum seeding capacity is $S^{max}_{a_n}$ if $ 1 < v_c^a < \bar v$, and is zero if $\bar v < v_c^a < v_h$.
\end{prop}
\begin{proof}
From condition \eqref{condition} the more agents with centralities above the threshold $v_c^{a}$, the more seeding budget can be allocated. Therefore, the maximum number of $k$ agents with centralities above the threshold $v_c^{a}$ must be found. Since $v_i \geq 1$ for all agents, first a centrality of $1$ is given to each agent and then the remainder of the centrality sum is distributed among maximum number of agents so that each agent receives at least $v_c^{a}-1$, making its overall centrality greater than $v_c^{a}$. It is easy to see that the number of such agents is upper bounded by $\lfloor \frac{\frac{2\alpha n}{2\alpha- \delta(1-\alpha)} - n}{v_c^a-1}\rfloor$.
This along with the fact that $1\leq k\leq n$ results in \eqref{k}.
Note that, in order to complete the proof, we should also provide an example achieving this maximum capacity. For $k=1$, the maximum seeding capacity is clearly achieved by the star graph with the seeding capacity of $S_{a_1}^{max}$. For $k\geq2$, a graph with largest seeding capacity is a $k$-star graph with $k$ central agents with the largest demand capacities and with equal centralities of
\begin{equation*}
v_h = \frac{n\delta (1-\alpha)}{k(2\alpha - \delta(1-\alpha))} + 1,
\end{equation*}
where $1 \leq k \leq n$ is given in \eqref{k}, and the remainder $n-k$ agents with the minimum centrality of $v_l=1$.
For the graph with minimum seeding capacity, similar to the proof of Proposition \ref{compare seeding}, we have minimum seeding capacity of $S^{max}_{a_n}$ in star graph if $ 1 < v_c^a < \bar v$, and zero in balanced graph if $ \bar v< v_c^a <v_h$.
\end{proof}

\begin{ex}
As a numerical example for the minimum and maximum seeding capacities, we consider a network with $n=15$ agents with demand capacity vectors of $S_a^{max}=S_b^{max}=(0.5)\mathbf{1}$, qualities of $q_a=q_b=1$, quality and seeding costs of $c_s=c_q=1$ and parameters of $\alpha = \delta = 0.5$. For this example from equations \eqref{condition} and \eqref{condition2} we have $v_c^{a}=v_c^{b}= 2.5$ and as a result, from equation \eqref{k} we get $k=3$. Therefore, a graph with the maximum seeding capacity is a $3$-star with seeding capacity of $1.5$ as illustrated in Fig.~\ref{fig:3star}. Also, since $ \bar v=\frac{4}{3} < v_c^a, v_c^b < v_h = 4.8 $, a balanced graph has the minimum seeding capacity of zero. A star graph has a seeding capacity of $0.5$ which is neither a minimum nor a maximum.
\end{ex}
\begin{figure}
\centering
\includegraphics[scale=2]{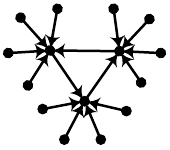}
\caption{A graph with maximum seeding capacity for Example~1}\label{fig:3star}
\end{figure}

As we saw, the structure of the graphs with minimum and maximum seeding capacity depends on the threshold value of $v_c^{a}$. However, for certain values of $v_c^{a}$ the seeding capacity will be independent of the structure of the graph, as described in the next proposition.
\begin{prop}
If $v_c^a > v_h$ no graph can be seeded. On the other hand, if $v_c^a < 1$ all graphs can be seeded equally up to agents' maximum demand capacities.
\end{prop}
\begin{proof}
The maximum possible centrality happens for the central agent of the star graph as shown in Lemma~\ref{v_Lv_H}. As a result, if $v_c^a > v_h$, then we have $v_c^a >  v_i$ for all $i$ in any graph and from condition \eqref{condition} no agent can be seeded. Also, since from definition $1 \leq v_i$ for all $i$, if $ v_c^{a} < 1$ then we have $ v_c^{a} < v_i$ for all agents in any graph, and any graph can be seeded up to agents' maximum demand capacities, given the availability of budget.
\end{proof}

\subsection{Effect of Quality of Products on Firms' Decisions:}
In this subsection we study the effect of the preset quality of each product, i.e. $q_a$ and $q_b$, on the optimal allocation of seeding and quality improvement budgets.

As it can be seen from equation \eqref{condition}, the threshold $v_c^{a}$ depends on both firm's and its rival's quality, i.e. both $q_a$ and $q_b$. In the next proposition we explain how the optimal seeding budget depends on $q_a$ and $q_b$.
\begin{prop} \label{prop quality}
Given a fixed graph, the optimal seeding budget is an increasing function of $q_a$. Furthermore, it is a decreasing function of $q_b$ if $q_b < q_a$ and an increasing function of $q_b$ if $q_a < q_b$.
\end{prop}
\begin{proof}
The optimal seeding budget is a decreasing function of the threshold value $v_c^{a}$ because if the threshold value $v_c^{a}$ increases then there will be less (or same) agents with centrality above the threshold which can be seeded in the optimal allocation. Therefore, the optimal seeding budget decreases. Also, $$ \frac{\partial v_c^{a}}{\partial q_a} = k \times \frac{-q_b}{(q_a+q_b)^3} <0,$$ for some constant $k > 0$. Therefore, the threshold value of $v_c^{a}$ is a decreasing function of $q_a$. This implies the first part of proposition. For the second part we have
$$ \frac{\partial v_c^{a}}{\partial q_b} = k \times \frac{(q_a-q_b)}{(q_a+q_b)^3},$$ for some constant $k > 0$. Therefore, the threshold value of $v_c^{a}$ is an increasing function of the quality of product $b$, when $q_b < q_a$ and a decreasing function of the quality of product $b$, when $q_b > q_a$. This completes the proof.
\end{proof}

Proposition \ref{prop quality} implies that higher quality in a firm's product results in a higher seeding budget in the optimal allocation. This can be due to the diminish return of quality: when quality is higher there is less need for quality improvement and it would be more profitable to spend on seeding. Furthermore, when $q_b < q_a$, the higher the quality of the rival firm's product, the lower the seeding budget of firm $a$, i.e., if $q_a\geq q_b'\geq q_b$ then, $\|S_a^*(q_b')\|_1\leq\|S_a^*(q_b)\|_1$. On the other hand, when competing with a firm whose product has a higher quality, i.e. $q_b>q_a$, the higher the quality of the rival firm's product, the higher firm $a$ should spend on seeding. In other words, if $q_a\leq q_b\leq q_b'$ then, $\|S_a^*(q_b)\|_1\leq\|S_a^*(q_b')\|_1$.

Combining these two results, we can see that given a fixed value of $q_a$, the seeding budget of firm $a$ is increasing with the difference $|q_a-q_b|$. The seeding budget attains its minimum when $q_b=q_a$, implying that the firm should allocate more budget to quality improvement to distance itself from the rival firm. However, as the gap between qualities widens, competition in qualities becomes less effective and firms spend more budget on seeding.
Proposition \ref{prop quality} demonstrates the effect of each quality on the optimal seeding budget when the other quality is fixed. In the next proposition, we compare the seeding budget of two firms in the optimal allocation of both firms.
\begin{prop}
Given an equal budget, the firm with higher quality also has higher seeding budget, i.e. if $q_a\leq q_b$, then $\|S_a^*\|_1\leq\|S_b^*\|_1$.
\end{prop}
\begin{proof}
From equations \eqref{condition} and \eqref{condition2} it can be easily seen that if $q_a\leq q_b$, then $v_c^{b}\leq v_c^{a}$. As a result, more agents satisfy the condition \eqref{condition2} for firm $b$ compared to firm $a$ and therefore, $\|S_a^*\|_1\leq\|S_b^*\|_1$.
\end{proof}

This result is due to diminish return of quality which means if a firm already has a high quality it would profit less by spending on quality improvement and it would be better for the firm to invest on seeding.

\subsection{Effect of Model Parameters on Firms' Decisions:}

In this subsection we study the effect of parameters of the model on the optimal allocation. There are two parameters: network coefficient $\alpha$, and discounting factor $\delta$, that affect the optimal strategy.
\begin{prop} \label{coef}
The optimal seeding budget is an increasing function of $\alpha$ and a decreasing function of $\delta$.
\end{prop}
\begin{proof}
As we saw earlier, the optimal seeding budget is a decreasing function of the threshold value $v_c^{a}$. Also,
since $\frac{1}{2} \leq \alpha \leq 1$ and $0 \leq \delta < 1$, it can be easily seen that numerator of $\lambda$ in equation \eqref{lambda} is decreasing in $\alpha$ and increasing in $\delta$ whereas denominator of the ratio is increasing in $\alpha$ and decreasing in $\delta$. Therefore, the threshold value $v_c^{a}$ is a decreasing function of $\alpha$ and an increasing function of $\delta$. This completes the proof.
\end{proof}

Preposition \ref{coef} implies that for a fixed graph and fixed qualities, if $\alpha$ is higher then seeding budget will be higher and quality improvement budget will be lower. This is because higher $\alpha$ means a lower coefficient for the graph, i.e. $(1-\alpha)$, which in turn implies that the effect of the graph and update dynamics is lower. Therefore, there is no point for firms to improve the quality and take advantage of the network structure to spread the product.
Moreover, if discounting factor $\delta$ is lower then seeding budget will be higher and quality improvement budget will be lower. This is because a low discounting factor indicates a short time horizon in utility of firms, therefore, it is more profitable to invest on seeding which is more effective in short term.

\subsection{Effect of Seeding and Quality Cost on Firms' Decisions:}
In this subsection we study the effect of the seeding and quality cost, i.e. $c_s$ and $c_q$, on the optimal allocation of seeding and quality improvement budgets.
\begin{prop} \label{costs}
The optimal seeding budget is a decreasing function of seeding cost $c_s$ and an increasing function of quality improvement cost $c_q$.
\end{prop}
\begin{proof}
The optimal seeding budget is a decreasing function of the threshold value $v_c^{a}$.
Also, we have
$$ \frac{\partial v_c^{a}}{\partial c_s} = k_1>0, \qquad \qquad  \frac{\partial v_c^{a}}{\partial c_q} = \frac{-k_2}{c_q^2} < 0, $$ for some constants $k_1, k_2 > 0$. Therefore, the threshold value of $v_c^{a}$ is an increasing function of the seeding cost $c_s$ and a decreasing function of the quality improvement cost $c_q$ and the result follows.
\end{proof}

Preposition \ref{costs} implies that for a fixed graph, fixed qualities and fixed parameters, if seeding cost $c_s$ is higher then the seeding budget will be lower and the quality improvement budget will be higher.
On the other hand, if quality improvement cost $c_q$ is higher then seeding budget will be higher and quality improvement budget will be lower.

\section{Conclusion} \label{sec5}

We proposed and studied a strategic model of marketing and product consumption in social networks. Two firms providing products with preset qualities, compete for maximizing the consumption of their products in a social network. Agents are myopic yet utility maximizing, given the quality of the products and actions of their neighbors. This myopic best response results in a local, linear update dynamics for the consumptions of the agents. At some point in time, firms spend a limited budget to marginally improve the quality of their products and to give free offers to a set of agents in the network in order to promote their products. We derived a simple threshold rule for the optimal allocation of the budget and described the resulting Nash equilibrium. We showed that the optimal allocation of the budget depends on the entire centrality distribution of the graph, quality of products, costs of seeding and quality improvement, discounting factor and network coefficient.
In particular, we showed that a graph with a higher number of agents with centralities above a certain threshold, has a higher seeding budget in the optimal allocation. Furthermore, if seeding budget is nonzero for a balanced graph, it will be nonzero for any other graph too, and if seeding budget is zero for a star graph, it will also be zero for any other graph. We also showed that firms allocate more budget to quality improvement when their qualities are close, in order to distance themselves from the rival firm. However, as the gap between qualities widens, competition in qualities becomes less effective and firms spend more budget on seeding.

\bibliographystyle{ieeetr}
\bibliography{CDC2014_REF}
\end{document}